\pdfoutput=1 
\documentclass[10pt,twocolumn]{IEEEtran}
\usepackage{amsmath,amsfonts,amssymb}
\usepackage{graphicx}
\usepackage{subfigure}
\usepackage{cite}

\newtheorem{theorem}{Theorem}
\newtheorem{definition}{Definition}
\newtheorem{lemma}{Lemma}

\newtheorem{proposition}{Proposition}

\newcommand{\set}[1]{\mathcal{#1}}
\newcommand{\Real}{{\mathbb{R}}}

\newcommand{\nequal}[1]{\stackrel{#1}{=}}

\newcommand{\graph}{\set{G}}
\newcommand{\nodes}{\set{P}}
\newcommand{\edges}{\set{E}}

\newcommand{\sourceLocation}{O}
\newcommand{\multicastRequirement}{M}
\DeclareMathOperator{\multicastProblem}{{\mathsf{T}}}
\newcommand{\destinationLocation}{\set{D}} 

\newcommand{\inputRate}{{\lambda}}
 
\newcommand{\edgeRate}{{\omega}}
\newcommand{\networkCoding}{{\Phi}} 
\newcommand{\sessions}{\set{S}}

\newcommand{\inputRV}{{T}}

\newcommand{\edgeRV}{{W}}

\def\edgeAlphabet{{w}}

\def\normal{r}
\def\groundset{\set{M}}
\def\espace{H}

\def\sRV{{S}}
\newcommand{\N}{\set{N}}

\title{Mission impossible: \\Computing the network coding capacity region}
\author{\authorblockN{Terence Chan and Alex Grant}\\
\authorblockA{Institute for Telecommunications Research\\
University of South Australia, Australia \\
{\tt \{terence.chan, alex.grant\}@unisa.edu.au }}}

\begin{document}
\maketitle

\begin{abstract}
  One of the main theoretical motivations for the emerging area of
  network coding is the achievability of the max-flow/min-cut rate for
  single source multicast. This can exceed the rate achievable with
  routing alone, and is achievable with linear network codes. The
  multi-source problem is more complicated.  Computation of its
  capacity region is equivalent to determination of the set of all
  entropy functions $\Gamma^*$, which is non-polyhedral. The aim of
  this paper is to demonstrate that this difficulty can arise even in
  single source problems. In particular, for single source networks
  with hierarchical sink requirements, and for single source networks
  with secrecy constraints. In both cases, we exhibit networks whose
  capacity regions involve $\Gamma^*$.  As in the multi-source case,
  linear codes are insufficient.
\end{abstract}

\section{Introduction}
Network coding~\cite{AhlCai00,LiYeu03} generalizes
routing by allowing intermediate nodes to perform coding operations
which combine received data packets.  One of the most celebrated
benefits of this approach is increased throughput in multicast
scenarios. This stimulated much of the early research in the
area. 
%
One fundamental problem in network coding is to understand the
capacity region and the classes of codes that achieve capacity. In the
single session multicast scenario, the problem is well understood. In
particular, the capacity region is characterized by max-flow/min-cut
bounds and linear network codes are sufficient to achieve maximal
throughput \cite{LiYeu03,DouFre05}. Network coding not only yields a
throughput advantage over routing, its capacity can be easily
determined, and easily achieved. This is in stark contrast to routing,
where computation of the capacity region and of optimal routes is
fundamentally difficult.

Significant practical and theoretical complications arise in more
general multicast scenarios, involving more than one session. An
expression for the capacity region is known~\cite{YanYeu07}, however
it is given by the intersection of a set of hyperplanes (specified by
the network topology and connection requirement) and the set of
entropy functions $\Gamma^*$.  Unfortunately, this capacity region, or
even the inner and outer bounds~\cite{SonYeu03,Yeu02,YeuLi06} cannot
be computed in practice, due to the lack of an explicit
characterization of the set of entropy functions for more than three
random variables. This difficulty is not simply a consequence of the
particular formulation of the capacity region given
in~\cite{YanYeu07}. It was recently shown that the problem of
determining the capacity region for the multi-source problem is in
fact entirely equivalent to the determination of $\bar{\Gamma}^*$, the
set of almost entropic functions
\cite{Chan.Grant07dualities}. Furthermore, the non-polyhedral nature
of $\bar{\Gamma}^*$, revealed in~\cite{Mat07} implies a non-polyhedral
capacity region (in contrast to the max-flow result for single
sources).  To make things even worse, it is also known that linear
network codes are not sufficient for the multi-source
problem~\cite{DouFre05,Chan.Grant07dualities}.

In this paper, we show that non-polyhedral capacity regions can occur
even in \emph{single source} scenarios. We demonstrate this phenomenon
for single source networks with hierarchical sink constraints, and for
single source networks with security constraints. Our approach is in
the spirit of our recent work \cite{Chan.Grant07dualities}, which
revealed a deep duality between network codes and entropy functions.
%
%
Direct consequences are
non-polyhedral capacity regions, the insufficiency of linear network
codes and the importance of non-Shannon information inequalities. 


Section \ref{sec:background} provides the basic setup for secure
network codes, and formally defines achievability and admissibility
for networks with wiretapping adversaries. Section \ref{sec:incr-mult}
focuses on the single source incremental multicast scenario, in which
the sinks have hierarchical requirements. Given a function $g$, we
construct an incremental multicast network that is solvable if and
only if $g$ is entropic. In Section \ref{sec:secure-multicast} we
construct a special single source secure multicast problem which is
equivalent to an insecure multi-source multicast problem. Invoking the
duality results from~\cite{Chan.Grant07dualities} these constructions
relate the solvability of both single-source incremental multicast and
single source secure multicast, to multi-source multicast problems.

\section{Background}\label{sec:background} 
 
The network topology will be modeled by a directed
acyclic graph $\set{G} = (\nodes, \set{E})$.  Vertices $u\in\nodes$
correspond to communication nodes and directed edges $e\in\set{E}$
are error-free point-to-point communication links. The
\emph{connection requirement}
 $\multicastRequirement \triangleq (\set{S},
\sourceLocation, \set{D})$ is specified by three components. The set
$\set{S}$ indexes the independent multicast sessions, each of which is
a collection of packets to be multicast to a prescribed set of
destinations.  The session-source location mapping
$\sourceLocation:\set{S}\mapsto\nodes$ specifies the originating node
$\sourceLocation(s)$ for session $s$.  The receiver-location mapping
$\set{D}:\set{S}\mapsto 2^{\nodes}$ indicates the set of nodes
$\set{D}(s)\subseteq\nodes$ which require the data of session $s$.



A \emph{network code} is identified by a set of discrete random
variables $\{T_{\set{S}}, W_{\set{E}}\}$, defined on finite sample
spaces, where for concise notation, set-valued subscripts denote a set
of objects indexed by the set, e.g. $Z_\set{X} =
\{Z_i,i\in\set{X}\}$. The source random variables $T_s, s\in\set{S}$
are mutually independent and are uniformly distributed on sample
spaces whose size will be denoted $|T_s|$. The variables $W_e,
e\in\set{E}$ are the messages transmitted over link $e$.

Since the network is acyclic, variables in $T_\set{S}$ and
$W_{\set{E}}$ can be ancestrally ordered according to the network
topology. Causal coding requires that edge messages are conditionally
independent of their non-incident ancestral messages given their
incident source and message variables.

\begin{definition}
  A network code is \emph{probabilistic} if there exists an outgoing
  link message which is not a function of the incoming source and link
  messages. Otherwise, it is \emph{deterministic}.
\end{definition} 

Probabilistic network codes can be implemented via using 
independent random variables $V_u$ (internal randomness) at each node
$u\in\nodes$ such that all outgoing messages from a node are
deterministic functions of incoming sources and link messages and the
independent randomness generated at the node. It is easy to prove that
all probabilistic network codes can be implemented in this
way. Accordingly, we shall specify a probabilistic network code by the
set $\{T_{\set{S}}, W_{\set{E}}, V_{\nodes}\}$.
\begin{lemma}\label{lemm:probcodeproperty}
  Given random variables $X_1,X_2$ and $V$, if $V$ is independent of
  $X_1$ and $X_2$, and $X_2$ is a function of $X_1$ and $V$, then
  $X_2$ is a function of $X_1$ alone
\end{lemma}
The implication of the lemma is as follows. At the sinks (or any
intermediate node) of the network, if reconstruction of the source
messages is possible, then it can also be achieved in the absence of
``internal randomness''.  In fact, in the absence of security
constraints, it is known that deterministic network codes are
sufficient \cite{Yeu02}. This is not always the case for the
wiretapping scenarios considered in Section
\ref{sec:secure-multicast}.

In addition to legitimate sinks, there are $|\set{R}|$ adversaries,
which can eavesdrop any message transmitted along a given collection
of links.
Each adversary attempts to reconstruct a particular set of source
messages, according to a wiretapping pattern.
\begin{definition}[Wiretapping pattern]
  The \emph{wiretapping pattern} is specified by a collection of tuples
  $(\set{A}_r, \set{B}_r)$ for $r\in\set{R}$ such that $\set{A}_r
  \subseteq \set{S}$ is the subset of sources to be reconstructed by
  adversary $r$, which observes only the links in $\set{B}_r$.
\end{definition}

For a given network code designed with respect to a connection
requirement $M$, define $P_e$ as the error probability that at least
one receiver fails to correctly reconstruct one or more of its
requested source messages.  A \emph{zero-error network code} is one
for which $P_e = 0$, and hence the source messages $T_{\set{S}}$ can
be perfectly reconstructed at desired sinks. The goal of secure
communications is to transmit information such that any eavesdropper
listening to the traffic on all the links in $\set{B}_r$ remains
``ignorant'' of the data transmitted by the sources in $\set{A}_r$. A
\emph{perfectly secure} network code is one for which the information
leakage ${I\left(T_{\set{A}_r} ; W_{\set{B}_r}\right)} = 0$ for all
$r\in\set{R}$.

\begin{definition}[Admissible rate-capacity tuple]\label{df:admissible}
  Given a network $\graph=(\nodes,\edges)$ and a connection
  requirement $\multicastRequirement$, a rate-capacity tuple
  $(\inputRate , \edgeRate)\triangleq (\lambda_{\set{S}},
  \omega_{\set{E}})$ is \emph{admissible} if there exists a
  \emph{perfectly secure}, \emph{zero-error} network code
  $\networkCoding = \{\edgeRV_f, f\in\sessions\cup \edges\}$, such
  that
  \begin{align*}
   H(\edgeRV_e) \le \log |\edgeAlphabet_e| \le \edgeRate_e,
  & \quad\forall e \in \edges,\\
   H(\inputRV_s)= \log |\edgeAlphabet_s| \ge \inputRate_s,
    &\quad\forall s\in \sessions,
  \end{align*}
  where $\edgeRV_e$ is the message symbol transmitted along link $e$
  and $\inputRV_s$ is the input symbol generated at source $s$.
\end{definition}


The preceding definitions consider zero-error network codes and
perfect security. Relaxing these requirements prompts the following
definition.
\begin{definition}[Achievable]\label{df:achievable}
  A rate-capacity tuple $(\inputRate , \edgeRate)$ is
  \emph{achievable} if there exists a sequence of network codes
  $\networkCoding^{(n)}$ and normalizing constants $\normal(n)>0$ such
  that
  \begin{align*}
    \lim_{n\to\infty} \frac{1}{\normal(n)} H\left(\edgeRV_e^{(n)}\right) \le
    \lim_{n\to\infty} \frac{1}{\normal(n)}\log |\edgeAlphabet_e^{(n)}|
    &\le
    \edgeRate_e,  \quad\forall e \in \edges,\\
    \lim_{n\to\infty} \frac{1}{\normal(n)}H\left(\inputRV_s^{(n)}\right) =
    \lim_{n\to\infty} \frac{1}{\normal(n)} \log
    |\edgeAlphabet_s^{(n)}| &\ge \inputRate_s, \quad\forall s\in \sessions, \\
    \lim_{n\to\infty} P_e\left(\networkCoding^{(n)}\right) =0, \\
\lim_{n\to\infty} \frac{1}{\normal(n)} I(T^{(n)}_{A_r} ; W^{(n)}_{B_r}) & = 0, \quad \forall r\in\set{R}.
\end{align*}
\end{definition}

In the absence of any security constraints, $|\set{R}|=0$, these
definitions reduce to the usual ones and the multi-source, multi-sink
capacity region is given by~\cite{YanYeu07}. Bounds for the
multi-source multi-sink scenario with wiretappers were given in
\cite{ChaGra08}.

\section{Incremental Multicast}\label{sec:incr-mult}
In this section, we study a the special case of \emph{incremental
  multicast}, meaning that the session indexes are totally ordered
such that a receiver requesting a particular session also requests all
sessions with lower index. We consider the simplest incremental
multicast scenario, with only two source messages and no secrecy
constraints (permitting deterministic codes).  We will show that
determining the capacity region, even in such a simple scenario, can
be no simpler than solving the general multicast problem.

Our approach is inspired by \cite{Chan.Grant07dualities}. Let
$\espace[\groundset]\subset\Real^{2^N}$ with coordinates indexed by
proper subsets of a ground set $\groundset$ with $N$ elements. Points
$h\in\espace[\groundset]$ can be regarded as functions,
$h:2^\groundset\mapsto\Real$ with $h(\emptyset)\triangleq 0$.
Given such an $h\in\espace[\groundset]$ we will construct a special
network $\graph^\dagger$, an incremental connection requirement
$\multicastRequirement^\dagger$ and a rate-capacity tuple
$\multicastProblem(h)$ that is admissible if and only if $h$ is 
entropic.

The network topology, connection requirement and link capacities are
defined in Figure \ref{fig:thenetwork}, which for convenience, is
divided into several subnetworks. The single source node is an open
circle, labelled with the two available sessions (this node is
repeated for convenience in Figures~\ref{fig:part1}, \ref{fig:lowerbd}
and~\ref{fig:butterfly}). The destinations are double circles,
labelled with their requirements. Intermediate nodes are solid
circles. The source and sink labels define the mappings
$\sourceLocation$ and $\destinationLocation$. Each capacitated edge is
labeled with a pair of symbols denoting the edge capacity, and the
edge message (and corresponding random variable).  Unlabelled edges
are assumed to be uncapacitated, or to have a finite but sufficiently
large capacity
to losslessly forward all received messages.

The first part of the network, shown in Figure \ref{fig:part1},
contains the source where there are two independent sessions (i.e.,
two messages $S_0$ and $S_1$) available. The desired source rates
associated with $S_0$ and $S_1$ are respectively
$\sum_{i\in\set{N}}h(i)$ and $h(\set{N})$.
 There are $2N$ specific edge messages
that are of particular interest. Rather than naming all edge variables
$\edgeRV_e, e\in\edges$, we label these $2N$ particular edge variables
$U_j$ and $V_j$ for $j=1,\dots,N$.  Remaining edge variables will be
labelled with generic symbols $W_i$ indexed by an integer $i$.
 
In Figure \ref{fig:part1}, the source node generates from $S_0$ and
$S_1$ respectively the sets of network coded messages $\{U_1, U_2,
\dots, U_N\}$ and $\{V_1, V_2, \dots, V_N\}$ which are duplicated as
required and forwarded to the rest of the network. The remainder of
the network is divided into subnetworks of two types, shown in Figures
\ref{fig:lowerbd} and \ref{fig:butterfly}.

\def\scalefactor{0.5}
\begin{figure}[htbp]
  \begin{center}
   \subfigure[Source node\label{fig:part1}]
    {\includegraphics*[scale=\scalefactor]{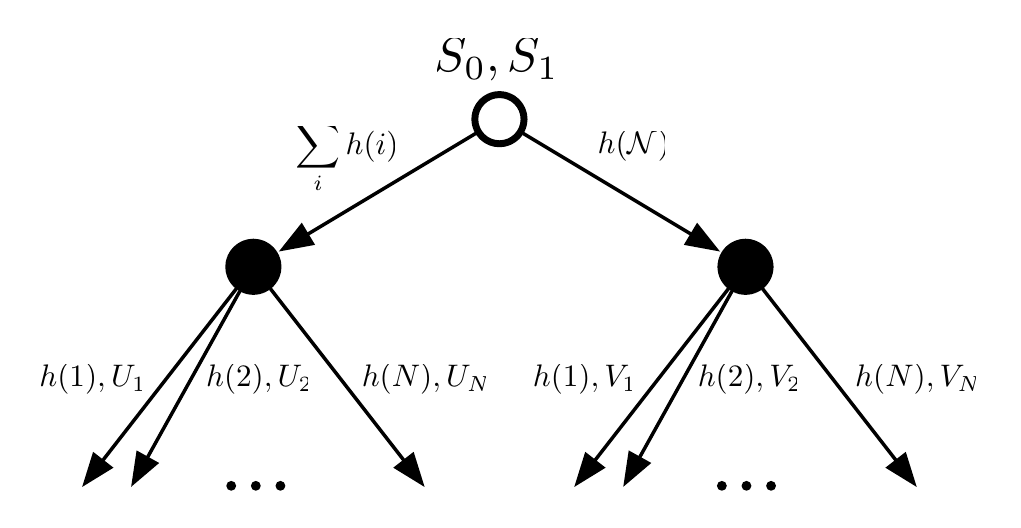}}
    \subfigure[{Type 1 subnetworks}\label{fig:lowerbd}]
    {\includegraphics*[scale=\scalefactor]{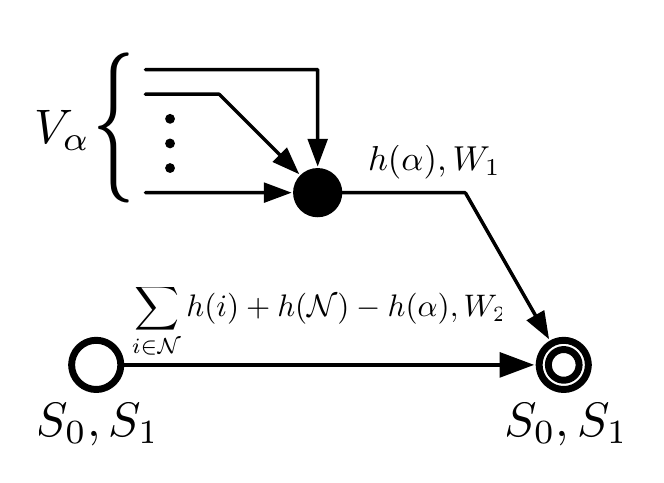}}

    \subfigure[{Type 2 subnetworks}\label{fig:butterfly}]
    {\includegraphics*[scale=\scalefactor]{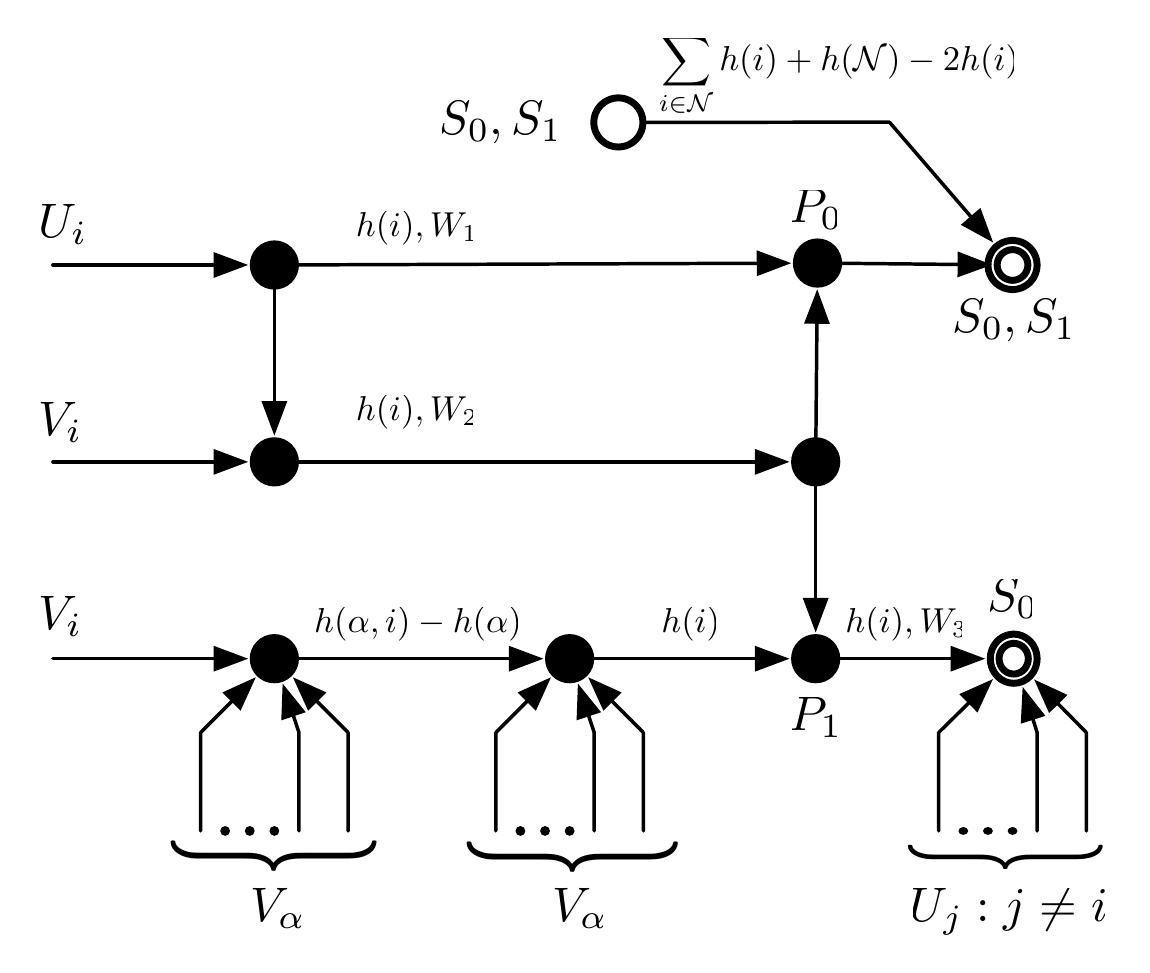}}
  \end{center}
  \caption{The network $\graph^\dagger$.}
  \label{fig:thenetwork}
\end{figure}

With reference to Figure \ref{fig:lowerbd}, there are $2^N-1$ type 1
subnetworks, one for each nonempty $\alpha\in 2^{\N}$. These
subnetworks introduce an edge of capacity $h(\N)-h(\alpha)$ between
the source and a sink requiring $S_1$. There is an intermediate node
which has another $|\alpha|$ incident edges (from Figure
\ref{fig:part1}), carrying $V_\alpha=\{V_j,j\in\alpha\}$.  The
intermediate node then has an edge of capacity $h(\alpha)$ to the
sink.

Figure \ref{fig:butterfly} shows the structure of type 2 subnetworks,
which are indexed by $\emptyset\neq\alpha\subset\N$ and an element
$i\in\N , i\not\in\alpha$.  Each type 2 subnetwork connects the source
to the upper receiver.  In addition, there are other incident edges
carrying $\{V_j:j\in \alpha\} $ and $\{U_j:j\in\set{N}\}$.  For
notational simplicity, we have written
$h\left(\alpha\cup\{i\}\right)\triangleq h(\alpha,i)$.

So far, we have described a network $\graph^\dagger$, a connection
requirement $\multicastRequirement^\dagger$ and have assigned rates to
sources and capacities to links.  Clearly
$\multicastRequirement^\dagger$ depends only on $N$, and not in any
other way on $h$. Similarly, the topology of the network
$\graph^\dagger$ depends only on $N$. The choice of $h$ affects only
the source rates and edge capacities, which are collected into the
rate-capacity tuple $\multicastProblem(h)$. Also, we can assume without loss
of generality that  $\multicastProblem(h)$ is a linear function of $h$.


\begin{definition}
  A function $h\in \espace[\N]$ is called \emph{entropic} if there
  exists discrete random variables $X_1,\dots, X_N$ such that the
  entropy of $\{X_i:i\in\alpha\}$ is equal to $h(\alpha)$ for all
  $\emptyset\neq\alpha\subseteq\set{N}$. Furthermore, $h$ is called
  \emph{quasi-uniform} if any subset of the variables are
  uniform over their support.
\end{definition}

\begin{theorem}\label{thm:direct}
  For the network $\graph^\dagger$ and a connection requirement
  $\multicastRequirement^\dagger$, if a rate-capacity tuple
  $\multicastProblem(h)$ is admissible, then $h$ is quasi-uniform and
  hence entropic.
\end{theorem}
\begin{proof}
  Suppose that $\multicastProblem(h)$ is admissible.  By
  Definition~\ref{df:admissible}, admissibility of
  $\multicastProblem(h)$ on $\graph^\dagger,
  \multicastRequirement^\dagger$ requires the existence of a
  zero-error network code $\networkCoding$ with source messages
  $\sRV_{[\alpha]}$, $\emptyset \neq \alpha\subseteq \N$ and a subset
  of its coded messages $U_\N$ and $V_\N$. Given this hypothesis, we
  will show that $h$ is the entropy function of $V_\N$, and that
  $V_\N$ is quasi-uniform.

  First focus on Figure~\ref{fig:part1}.  Applying min-cut bounds, it
  is straightforward to prove
  \begin{align*}
    H(U_\N , V_\N) &= \sum_{i\in\set{N}} H(U_i) + H(V_\N),\\
    H(U_i) &= h(i), \forall i\in\set{N},\\
    H(V_\N)&= h(\set{N}).\\
    H(V_i)&= h(i), \forall i\in\set{N}.
  \end{align*}
  Similarly, applying min-cut bounds to type 1 subnetworks of
  Figure~\ref{fig:lowerbd}, $ H(V_\alpha) \ge h(\alpha),
  \emptyset\neq\alpha\subseteq\set{N}$.

  We now focus on type 2 subnetworks of Figure~\ref{fig:butterfly} and
  aim to prove that $H(V_\alpha) \le h(\alpha) $ for any
  $\emptyset\neq\alpha\subseteq\set{N}$. In order for the upper
  receiver to reconstruct $S_0$ and $S_1$,
  \begin{align*}
    H(W_1,W_2) + h(\set{N})+  \sum_{j\neq i} h(j)  - 2h(i)   & \ge H(S_0,S_1)
  \end{align*}
  or equivalently, $H(W_1,W_2) \ge 2h(i)$.  In addition,
  \begin{align*} 
    H(W_1,W_2) & \le H(U_i,V_i, W_1,W_2) \\
    & = H(U_i, V_i) \le  2h(i) .
  \end{align*}
  As a result, $H(W_1,W_2) = H(U_i,V_i, W_1,W_2)$ which further
  implies that $V_i$ is a function of $W_1,W_2$. Thus $V_i$ can be
  recovered at $P_0$.  On the other hand, from the lower part of the
  subnetwork,
  \begin{align*}
    H(U_i | W_3) & = H(U_i | W_3, U_j, j\neq i) + I(U_i;U_j, j\neq i | W_3) \\
    &\nequal{(a)} I(U_i;U_j, j\neq i | W_3) \\
    & \le  I(U_i,W_3 ; U_j, j\neq i)  = 0  
  \end{align*}
  where $(a)$ follows from the fact that $S_0$ can be reconstructed at
  the lower receiver. This implies that $U_i$ can be reconstructed at
  $P_1$. From~\cite{Chan.Grant07dualities}, that $P_0$ can decode
  $V_i$ and that $P_1$ can decode $U_i$ further implies
  $H(V_i|V_\alpha) = h(\alpha,i)-h(\alpha)$.  By mathematical
  induction (similar to the proof of \cite[Theorem
  1]{Chan.Grant07dualities}), the only solution that satisfies all of
  the conditions above is when the entropy function of $V_\N$ is equal
  to $h$.

  Finally, from type 1 subnetworks, the support of $V_\alpha$ is at
  most $2^{h(\alpha)}$. Hence, $V_\alpha$ is indeed quasi-uniform
  (this also implies that the $U_i$ are quasi-uniform, via
  $H(U_i)=H(V_i)=h(i)$ and the independence of the $U_i$).
\end{proof}

\begin{theorem}[Converse]\label{thm:converse}
  For the network $\graph^\dagger$ and a connection requirement
  $\multicastRequirement^\dagger$, a rate-capacity tuple
  $\multicastProblem(h)$ is admissible if $h$ is quasi-uniform.
\end{theorem} 

From Theorems \ref{thm:direct} and \ref{thm:converse}, we can follow
the approach in \cite{Chan.Grant07dualities} and easily extend the
result to almost entropic functions.
\begin{theorem}\label{thm:almostent}
  For the network $\graph^\dagger$ and a connection requirement
  $\multicastRequirement^\dagger$, a rate-capacity tuple
  $\multicastProblem(h)$ is  achievable if and only if
  $h$ is almost entropic\footnote{A function $h$ is almost entropic if
    it is the limit of a sequence of entropic functions.}.
\end{theorem}

\section{Secure Multicast}\label{sec:secure-multicast}
Linear network codes (for single source multicast) that are resilient
to eavesdropping are considered in \cite{cai2002snc}. Sufficient
conditions for the existence of such codes was also derived. This was
further generalized in \cite{CaiYeu07} to multi-source cases.  A
similar result was also obtained in \cite{feldman2004csn} which gives
necessary and sufficient conditions under which transmitted data are
safe from being revealed to eavesdroppers. All of the above-cited
works assume that the wiretapper aims to reconstruct all sources.
Similar results have been obtained where only a subset of sources are
to be reconstructed~\cite{bhattad2005wsn}. Inner and outer bounds to
the secure capacity region were given in \cite{ChaGra08}.

We will now show that even for a simple single-session secure
multicast problem, determination of the capacity region can be
extremely hard.  In particular, the problem is at least as hard as any
multi-source multi-session multicast problem.

Figure~\ref{fig:gstar} shows the construction for a network
$\graph^\star$. The source message is $X$ whose  rate is
$d$. The link capacities are parametrized by $0 < c < d$.  There is a
single eavesdropper who only observes the message variable
$W_3$. Thus Figure~\ref{fig:gstar} also specifies
$\multicastRequirement^\star$, and the wiretapping pattern
$\set{A}^\star, \set{B}^\star$.

\begin{figure}[htb]\centering
  \includegraphics[scale=\scalefactor]{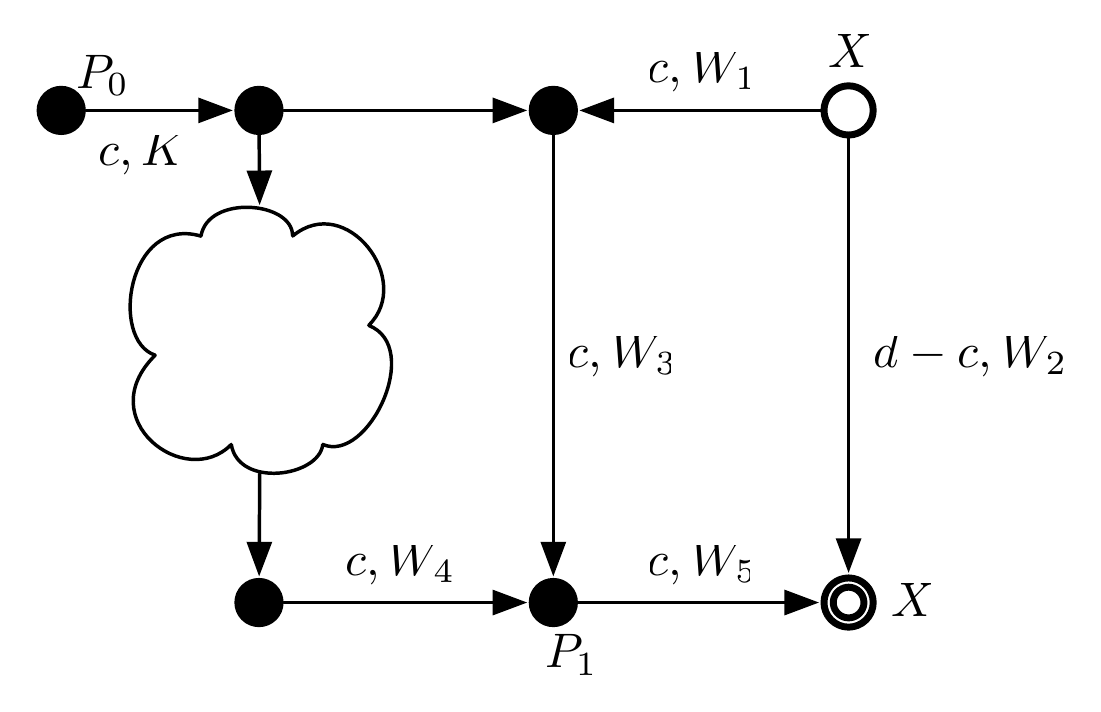} 
  \caption{The network $\graph^\star$.}\label{fig:gstar}
\end{figure}

\begin{proposition}\label{prop:decodeablekey}
  Given network $\graph^{\star}$ and  connection (and secrecy)
  requirement $\multicastRequirement^{\star}$ depicted in Figure
  \ref{fig:gstar}, if a rate-capacity tuple $\multicastProblem(h)$ is
  admissible then $K$ is a function of $W_4$.
\end{proposition}
\begin{proof}
  From the capacity constraint on $\graph^{\star}$, we
  have
  \begin{align*}
    H(W_1,W_2)& \le H(W_1)+H(W_2) \\
    & = c + d -c \\
    & = H(X).
  \end{align*}
  Together with the decodability requirement, $H(X|W_1,W_2) = 0$, we
  have
  \begin{align*}
   H(W_1, W_2)&=H(W_1)+H(W_2)  \\
   H(W_1,W_2|X)&=0 \\
   H(W_1,W_2) &= H(X) \\
   H(W_1) &= c \\
   H(W_2)&=d-c.
  \end{align*}
  Applying a min-cut bound on the set of edge variables $\{W_2,
  W_5\}$, we can also prove that $H(W_5) = c$ and $H(W_5|X) = 0$.  On
  the other hand, the secrecy constraint requires $I(W_3;X)=0$ and
  hence
  \begin{equation}\label{eqn:a}
    I(W_1;W_3)=0 
  \end{equation}
  as $W_1$ is a function of $X$. 


 Now, we will show that $H(W_1 | W_3,W_4 ) = 0$.  First,
  \begin{align*}
    I(W_3,W_4;X) & \nequal{(a)} I(W_3,W_4;W_1,X)  \\
    &= I (W_3,W_4;W_1) + I(W_3,W_4;X|W_1) \\
    &\nequal{(b)} I (W_3,W_4;W_1)      
  \end{align*}
  where (a) follows from the fact that $W_1$ is a function of $X$ and
  (b) follows from the conditional independence implied by the
  underlying network topology.  Using the same argument, we can also
  prove that $  I(K,W_3; X) =  I (K,W_3;W_1)  $.

  Since $W_5$ is a function of $X$ and is thus independent of internal
  randomness, Lemma \ref{lemm:probcodeproperty} implies that
  $H(W_5|W_3,W_4) = 0$.  Together with $H(W_5)=c$, we have
  \begin{align*}
    I (W_3,W_4;W_1)  & =  I (W_3,W_4;X )  \\
    & \ge I (W_3,W_4;W_5 ) \\
    & = H(W_5) = c.
  \end{align*}
  Since $H(W_1) = c$, it implies that $H(W_1|W_3,W_4)=0$ or
  equivalently that $W_1$ is a function of $W_3$ and $W_4$. Similarly,
  using the same argument, we can also prove that $H(W_1| K,W_3)=0$.

  Our final aim is to show that $H(K)=H(W_4)=c$ and $H(K,W_4) = 2c$.
  Clearly, both $H(K)$ and $H(W_4)$ are bounded above by $c$ due to
  the edge capacity constraint. We obtain a lower bound on the entropy
  of $K$ as follows.
  \begin{align*}
    H(K) & \ge I(K;W_1 | W_3) \\
    & =  I(K;W_1 | W_3) + H(W_1|K,W_3) \\
    & = H(W_1|W_3) \\
    & \nequal{(a)} H(W_1) = c
  \end{align*}
  where (a) follows from (\ref{eqn:a}). Hence, $H(K)=c$. And
  similarly, we can also prove that $H(W_4)=c$.

  Independence of $W_1$ and $W_3$ implies
  \begin{align*}
    H(K|W_1,W_3) & = H(W_1,K,W_3) - H(W_1,W_3) \\
    &= H(K,W_3) - H(W_1,W_3) \\
    &= H(K,W_3) - H(W_1) - H(W_3) \\
    &\nequal{(a)}  H(W_3|K) - H(W_3)  \le 0,
  \end{align*}
  where (a) follows from $H(W_1) = H(K)= c$.  Consequently,
  $H(K|W_1,W_3)=0$.

  Similarly, $H(W_4|W_1,W_3) = 0$. Finally,
  \begin{align*}
    2c & \ge H(W_1,W_3) \\
    & = H(W_1,K,W_3,W_4) \\
    & \ge   H(W_1, K, W_4) \\
    & \nequal{(a)} H(W_1) + H(K,W_4)    \ge 2c   
  \end{align*}
  where (a) follows from independence of $W_1$ and $(K,W_4)$.
  Hence, $H(K,W_4) = c$ which further implies  $H(K|W_4) = H(W_4|K)=0$. 
\end{proof}

Under a regularity condition (that $2^c$ and $2^d$ are integers), the
converse of Proposition~\ref{prop:decodeablekey} also holds.
\begin{proposition}[Converse]\label{prop:decodeablekeyconv}
  For the network $\graph^{\star}$ with connection (and secrecy)
  requirement $\multicastRequirement^{\star}$, the specified
  rate-capacity tuple is admissible if a secret key of a rate $c$ can
  be transmitted from the node $P_0$ to $P_1$.
\end{proposition}
 
Essentially, Propositions \ref{prop:decodeablekey} and
\ref{prop:decodeablekeyconv} suggest that the admissibility of the
single source secure multicast problem depends on communication of a
secret key from $P_0$ to $P_1$.  Adhering several copies of
$\graph^\star$ together (see Figure~\ref{fig:manygstar}), we can
easily generalize the network such that admissibility implies that
multiple secret keys must be transmitted across a network. This turns
the single source secure multicast problem into a multi-source
multicast.


\begin{theorem}\label{thm:securemulticast}
  For any multicast problem (without secrecy constraints), there exists
  a corresponding secure multicast problem such that the multicast
  problem is admissible if and only if the corresponding secure
  multicast problem is also admissible. 
Consequently, using the single-source two sessions network 
  $\graph^\dagger$ and a connection requirement
  $\multicastRequirement^\dagger$, there exists a secure multicast problem
  such that a rate capacity tuple  $\multicastProblem(h)$ is  achievable if and only if
  $h$ is almost entropic.
\end{theorem}

\begin{figure}[htb]\centering
  \includegraphics[scale=0.4]{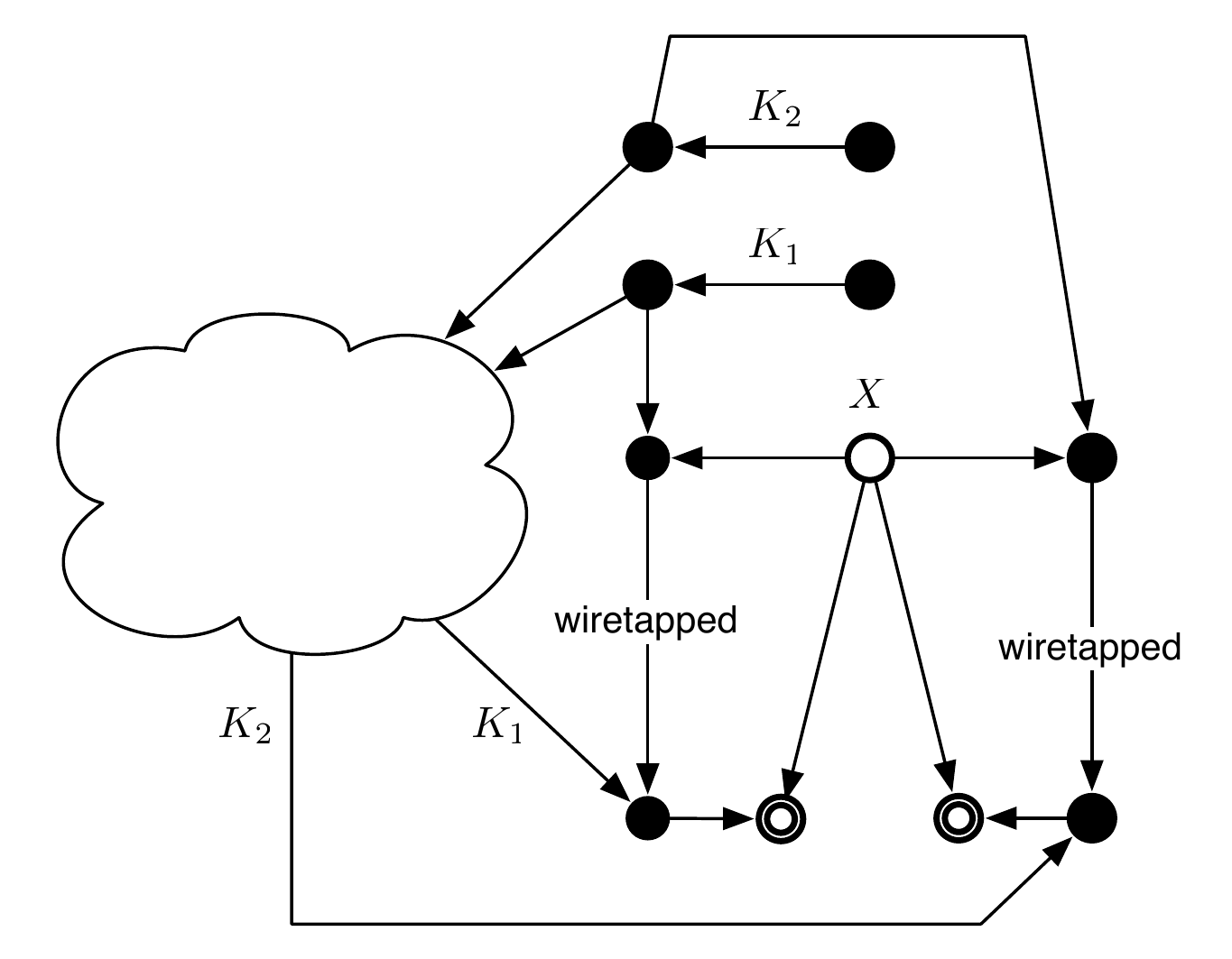} 
  \caption{Several copies of $\graph^\star$.}\label{fig:manygstar}
\end{figure}

\section{Implications and conclusion}\label{sec:implications}

Theorems \ref{thm:almostent} and \ref{thm:securemulticast} show that
even for a single-source network multicast problem with two
independent sets of messages or for a single source secure multicast
problem, the determination of the set of achievable rate-capacity
tuples can be extremely hard.  Following the same arguments as used
in~\cite{Chan.Grant07dualities}, we can also prove the following
results for a single-source two-session multicast problem or for a
single-source single-session multicast problem with secrecy
constraints:

\begin{enumerate}
\item Capacity regions are not polyhedral\footnote{That the
    single-source single-session secure multicast problem has a
    non-polyhedral capacity region is somewhat surprising, since the
    region for the same problem without the secrecy constraint is
    completely determined by the min-cut bound} in general.
\item LP bounds are not tight in general.
\item Linear codes are not sufficient to achieve capacity.
\end{enumerate}
In other words, finding capacity regions for (secure) multicast
problems seems to be a mission impossible. Not only are the existing
bounding techniques loose, the non-polyhedral nature of the capacity
region suggests that LP bounds cannot fully characterize the region,
even with the addition of more and more newly discovered information
inequalities. Any finite set of such new inequalities can only further
tighter the bound, but can never yield the exact capacity region.

Despite the hardness of the problem, there are still many questions to
be answered. It is unclear what makes finding the capacity region
problem so difficult.  In the case of a single session multicast or
the case where there are only two sinks, capacity regions have
explicit polyhedral characterizations provided by min-cut bounds. On
the other hand, where there are many sinks, the capacity region can be
extremely complicated to characterize, even if there are only two
independent
sessions. 
It will be of great importance to classify the set of networks and
connection requirements that lead to polyhedral capacity regions
characterized by min-cut bounds or LP bounds.




\end{document}